\documentclass[twocolumn,pra]{revtex4-2}
\usepackage[dvips]{graphicx} 
\usepackage{amsfonts,amscd,amsmath,amsthm}
\usepackage{enumerate}
\usepackage{epsfig}
\usepackage{subfigure}
\usepackage{xcolor}
\usepackage[colorlinks = true]{hyperref}
\usepackage{physics}
\usepackage{epstopdf}
\usepackage{framed}
\usepackage{multirow}
\usepackage{color}
\usepackage{longtable}
\usepackage{comment}
\usepackage[ruled,vlined]{algorithm2e}
\usepackage[most]{tcolorbox}

\graphicspath{{./figure/}}

\usepackage{tikz}
\usetikzlibrary{tikzmark, calc, fit, positioning}
\usetikzlibrary{shapes,arrows}
\usetikzlibrary{fadings,snakes}
\usetikzlibrary{decorations.pathmorphing,patterns}
\usetikzlibrary{quantikz}

\newtheorem{lemma}{Lemma}

\newtheorem{observation}{Observation}

\newtcolorbox[auto counter]{mybox}[2][]{
	enhanced,
	breakable,
	colback=blue!5!white,
	colframe=blue!75!black,
	fonttitle=\bfseries,
	title=Box \thetcbcounter: #2,#1
}

\begin{document}

\title{Approximate quantum error correction, covariance symmetry, and their relation}

\author{Hao Dai}
\email{dhao@mail.tsinghua.edu.cn}
\affiliation{Center for Quantum Information, Institute for Interdisciplinary Information Sciences, Tsinghua University, Beijing 100084, P.~R.~China}
\affiliation{Hefei National Laboratory, University of Science and Technology of China, Hefei, Anhui 230088, P.~R.~China}

\begin{abstract}
     To perform reliable quantum computation, quantum error correction is indispensable. In certain cases, continuous covariance symmetry of the physical system can make exact error correction impossible. In this work we study the approximate error correction and covariance symmetry from the information-theoretic perspective. For general encoding and noise channels, we define a quantity named infidelity to characterize the performance of the approximate quantum error correction and quantify the noncovariance of an encoding channel with respect to a general Lie group from the asymmetry measure of the corresponding Choi state. In particular, when the encoding channel is isometric, we derive a trade-off relation between infidelity and noncovariance. Furthermore, we calculate the average infidelity and noncovariance measure for a type of random code. 
\end{abstract}

\maketitle

\section{introduction}

Errors are inevitable in quantum computing and quantum error correction (QEC) provides a method to realize fault-tolerant quantum computation \cite{gottesman2010introduction,campbell2017roads}. The subject has been studied for decades and various correcting codes have been developed \cite{RevModPhys.87.307,Devitt2013quantum}. Beyond the quantum computation, QEC is closely connected with a wide range of quantum topics, such as quantum metrology \cite{PhysRevLett.112.080801,PhysRevLett.116.230502,Shettell_2021} and quantum entanglement \cite{PhysRevA.54.3824,brun2006correcting,verlinde2013black}.

Symmetry is a ubiquitous property of the physical system and can put strong constraints on the QEC. A no-go theorem, also known as the Eastin–Knill theorem, claims that there does not exist a local error-detecting code in a finite-dimensional system that allows for a set of universal logical gates to act transversally on the physical system \cite{PhysRevLett.102.110502}. This theorem implies that the continuous covariance symmetry and exact correction can be incompatible in certain cases \cite{PRXQuantum.2.010326,PhysRevX.10.041018}, which has motivated the exploration of the relation between covariance symmetry and approximate QEC. Several studies have focused on the performance of quantum codes that are exactly covariant but correct errors approximately \cite{PhysRevLett.126.150503,PRXQuantum.3.020314,PhysRevResearch.4.023107,Zhou2021newperspectives}. In particular, when the symmetry group is the U(1) Lie group the corresponding generator in the physical system is a Hamiltonian, covariant codes cannot correct errors perfectly if the physical Hamiltonian satisfies the Hamiltonian-in-Kraus-span (HKS) condition\cite{liu2021approximate,PRXQuantum.2.010343}. Under this special case, the relation between the covariance violation and the inaccuracy of the approximate QEC has been investigated \cite{liu2021quantum,liu2021approximate}. 

In this work we study the approximate QEC and the covariance symmetry from an information-theoretic perspective. For general encoding and noise channels in the form of the Kraus representations, we evaluate the error-correcting capability of the codes via a defined quantity called infidelity, which is related to entanglement fidelity. When infidelity is equal to $0$, the errors caused by the noise channel can be corrected exactly. We also quantify the violation of covariance symmetry, which we term noncovariance, from the asymmetry measures of the corresponding Choi state. We specifically explore the infidelity and noncovariance measure for isometric encoding codes. Moreover, we prove again that under the HKS condition, exact correctability and covariance are incompatible. In addition, we investigate the generalized Wigner-Yanase skew information and derive a sum uncertainty relation. By virtue of the generalized skew information, we obtain a trade-off
relation between infidelity and noncovariance. Furthermore, we also calculate the average infidelity and noncovariance measure for a type of random code.

The paper is organized as follows. In Sec. \ref{sec:pre} we review the basic concepts including QEC, Wigner-Yanase skew formation and asymmetry measures for states. In Secs. \ref{sec:app} and \ref{sec:nonc} we quantify the inaccuracy of the approximate QEC and the noncovariance, respectively. In Sec. \ref{sec:iso} we study the special case for the isometric encoding channel. In Sec. \ref{sec:out} we summarize and offer a suggestion for future work.


\section{Preliminaries}\label{sec:pre}
In this section, to highlight the idea of our approach, we briefly review the basic working knowledge and clarify some notation.

\subsection{Quantum error correction}

 In a QEC procedure, the logical state is encoded into a higher-dimensional physical system and redundancy is introduced to protect against errors. As a starting point, we denote by $L$ the logical system and by $\mathcal{H}_L$ the relevant Hilbert space. The dimension of the Hilbert space is assumed to be $d_L$ and the state space is denoted by $\mathcal{D}(\mathcal{H}_L)$.  Similar definitions can be defined for other systems. The encoding is a channel $\mathcal{E}$ from the logical system $L $ to the physical system $S$. The subspace of system $S$, $\mathcal{C}=\mathcal{E}(\mathcal{D}(\mathcal{H}_L))$, is known as code space and the projector on the code space is denoted by $P$. After a noise channel $\mathcal{N}(\rho)=\sum_{i=1}^{n}A_i\rho A_i^{\dagger}$ with $\sum_{i=1}^{n}A_i^{\dagger} A_i=\mathbf{1}$, 
the encoding state is changed and we can perform a corresponding decoding channel$\mathcal{R}$ to recover the original state. An ideal QEC procedure can recover all states in the logical system perfectly, that is,
\begin{equation}
    \mathcal{R} \circ \mathcal{N} \circ \mathcal{E}=\mathcal{I}
\end{equation}
with $\mathcal{I}$ being the identity map on logical system $L$. 


The Knill-Laflamme condition is a necessary and sufficient condition for a quantum code to achieve an exact correction \cite{PhysRevA.55.900}. For a given code $\mathcal{E}$ with the projector on the code subspace $P$, the errors can be corrected if and only if
\begin{equation}
    P A_i^{\dagger} A_j P=\alpha_{ij} P
\end{equation}
holds for a corresponding non-negative Hermitian matrix $(\alpha_{ij})$. Note that when the Kraus operators of a noise channel can be described by a linear span of $\{A_i\}$, the errors caused by
this noise can also be corrected exactly.

\subsection{Wigner-Yanase skew information and its generalization}
The conventional variance quantifies the total uncertainty of the observable $H$ in the state $\rho$ and is defined as
\begin{equation}
    V(\rho,H)=\tr(\rho H^2)-(\tr\rho H)^2.
\end{equation}
As a counterpart to it, the quantity
\begin{equation}\label{eq:skew}
    I(\rho,H)=-\frac{1}{2}\tr [\sqrt{\rho},H]^2=\frac{1}{2}\norm{[\sqrt{\rho},H]}^2_2
\end{equation}
also known as the Wigner-Yanase skew information \cite{wigner1963information,lieb1973convex,luo2004skew}, can quantify the quantum uncertainty of the observable $H$ in the state $\rho$. Here $[X,Y]=XY-YX$ is the Lie product and $\norm{X}_p=(\tr(XX^{\dagger})^{p/2})^{1/p}$ is the $p$-norm. For a pure state, the skew information coincides with the variance. 


The operator $H$ in Eq.~\eqref{eq:skew} is required to be Hermitian and we can generalize to non-Hermitian case \cite{PhysRevA.98.012113}. For an arbitrary operator $K$ which can be non-Hermitian, the generalized skew information is defined as
\begin{equation}
    I(\rho,K)=\frac{1}{2}\tr[\sqrt{\rho},K][\sqrt{\rho},K]^{\dagger}=\frac{1}{2}\norm{[\sqrt{\rho},K]}^2_2.
\end{equation}
In particular, for a pure state $\ket{\phi}$, 
\begin{equation}
    I(\ketbra{\phi},K)=\frac{1}{2}\bra{\phi}KK^{\dagger}+K^{\dagger}K\ket{\phi}-\abs{\bra{\phi}K\ket{\phi}}^2.
\end{equation}
In addition, the generalized skew information can be expressed as a sum of the original skew information
\begin{equation}\label{eq:reim}
    I(\rho,K)=I(\rho,K^{\dagger})=I(\rho,\Re(K))+I(\rho,\Im(K)),
\end{equation}
where $\Re(K)=\frac{1}{2}(K+K^{\dagger})$ and $\Im(K)=\frac{1}{2i}(K-K^{\dagger})$ represent the real and imaginary components, respectively. Note that when $K$ is Hermitian, the generalized skew information degenerates to the original ones.

The original skew information satisfies a series of uncertainty relations \cite{chen2016sum,zhang2021tighter,cai2021sum}. Here we give a sum uncertainty relation based on the generalized skew information.

\begin{lemma}\label{lem}
Let $K_1,\cdots, K_N$ be a set of operators. For a state $\rho$, there is
\begin{equation}
   \sum_{j=1}^{N} I(\rho,K_j)\geq \frac{1}{N}I(\rho,\sum_{j=1}^{N} K_j).
\end{equation}
\end{lemma}
\begin{proof}
\begin{equation}
    \begin{split}
    I(\rho,\sum_{j=1}^N K_j)&=\frac{1}{2}\norm{[\sqrt{\rho},\sum_{j=1}^N K_j]}_2^2\\
    &\leq \frac{1}{2}\Big(\sum_{j=1}^N\norm{[\sqrt{\rho},K_j]}_2\Big)^2\\
    &\leq\frac{N}{2} \sum_{j=1}^N\norm{[\sqrt{\rho},K_j]}_2^2 \\
    &= N \sum_{j=1}^N I(\rho,K_j),
    \end{split}
\end{equation}
where the first inequality is from the triangle inequality of the norm and the second inequality
is from the Cauchy-Schwarz inequality.
\end{proof}

\subsection{Asymmetry measures}

Given a group $\mathbf{G}$, for any group element $g$, let $U(g)$ be the unitary operator represented in the space $\mathcal{H}$. If a state $\rho$ remains unchanged under unitary transformations induced by $\mathbf{G}$,
\begin{equation}\label{eq:symsta}
    U(g)\rho U^{\dagger}(g)=\rho, \quad \forall g\in \mathbf{G},
\end{equation}
we say that the state is symmetric with respect to $\mathbf{G}$. In quantum resource theory, the quantification of how much a state breaks this symmetry, or the measure of asymmetry, is a significant problem. Different measures of asymmetry have been proposed in the literature \cite{marvian2014extending,PhysRevA.90.062110,PhysRevA.94.052324,Li_2020}. For example, some commonly used measures of asymmetry are based on skew information and von Neumann entropy. Here we mainly focus on Lie groups and only review the asymmetric measure given by skew information. Suppose the Lie algebra of the Lie group $\mathbf{G}$ has an orthonormal base $\{H_p:p=1,\cdots,d_{\mathbf{G}}\}$, where $d_{\mathbf{G}}$ is the dimension of the Lie algebra. All generators can be written as linear combinations of the elements in this base. The sum of the skew information,
\begin{equation}
    N_{\mathbf{G}}(\rho)=\sum_{p=1}^{d_{\mathbf{G}}} I(\rho,H_p),
\end{equation}
quantifies the asymmetry of the state $\rho$ with respect to the group $\mathbf{G}$ \cite{marvian2014extending,Li_2020}. The asymmetry measure possesses the following desirable properties.

(i) Here $N_{\mathbf{G}}(\rho)\geq 0$ and the equality holds if and only if the state commutes with all generators, which indicates that the state is symmetric with respect to $\mathbf{G}$.

(ii) For all $ g$, there is 
    \begin{equation}
        N_{\mathbf{G}}(\rho)=N_{\mathbf{G}}(U(g)\rho U^{\dagger}(g)).
    \end{equation}

(iii) Here $N_{\mathbf{G}}(\rho)$ is convex in the sense that 
    \begin{equation}
        N_{\mathbf{G}}(\sum_i \lambda_i \rho_i)\leq \sum_i \lambda_i N_{\mathbf{G}}(\rho_i),
    \end{equation}
    where $\lambda_i\geq 0$ and $\sum_i \lambda_i=1$.

Items (i) and (iii) can be directly deduced from the properties of the skew information \cite{Li_2020}. We only need to prove item (ii). For all $g$, from the unitary invariance of the skew information, we can obtain
\begin{equation}
   I(U(g)\rho U^{\dagger}(g),H_p)=I(\rho , U^{\dagger}(g) H_p  U(g)).
\end{equation}
For any unitary operator $U(g)$ and the generator $H_p$, $U^{\dagger}(g) H_p  U(g)$ is also a generator in Lie algebra \cite{hall2013lie}. Consequently, $\{U^{\dagger}(g) H_p  U(g):p=1,\cdots, d_{\mathbf{G}}\}$ forms an orthonormal base of the Lie algebra. Since the sum of the skew information does not depend on the choice of the orthonormal base \cite{PhysRevA.73.022324}, item (ii) holds.

\section{approximate quantum error correction}\label{sec:app}

The exact correctability is a strong restriction to practical codes. As a result, we consider approximate QEC codes in some cases and many quantifiers have been developed to evaluate the performance of the approximate error correction\cite{Zhou2021newperspectives,PRXQuantum.3.020314,PhysRevA.71.062310,PhysRevA.81.062342,PhysRevA.99.022313,PhysRevA.107.032422}. In an approximate QEC process, we need to find a proper recovery channel such that the composite operation $\mathcal{R} \circ \mathcal{N} \circ \mathcal{E}$ is close enough to the identity map, which demonstrates that all states can be nearly recovered. To characterize the performance of the approximate QEC codes, we first recall how to quantify the ``distance'' between two channels.

For two states $\rho$ and $\sigma$ the distance is quantified by fidelity
\begin{equation}
    F(\rho,\sigma)=\norm{\sqrt{\rho},\sqrt{\sigma}}_1=\tr\sqrt{\sqrt{\rho}\sigma\sqrt{\rho}}.
\end{equation}
The fidelity and the trace distance are closely related. The two measures are qualitatively equivalent since they satisfy the inequalities \cite{nielsen2002quantum},
\begin{equation}\label{eq:ineqf}
    1-\frac{1}{2}\norm{\rho-\sigma}_1\leq F(\rho,\sigma)\leq \sqrt{1-\frac{1}{4}\norm{\rho-\sigma}^2_1}.
\end{equation}
For two channels $L$, $\Lambda$ and $\Lambda'$ in the system, the entanglement fidelity 
\begin{equation}
\begin{split}
&F_e(\Lambda,\Lambda')\\
&=F\Big((\Lambda_L\otimes \mathcal{I}_R)(\ketbra{\psi}_{LR}),(\Lambda'_L\otimes \mathcal{I}_R)(\ketbra{\psi}_{LR})\Big)
\end{split}
\end{equation}
measures the closeness between these two channels, where $R$ is the reference system identical to system $L$ and $\ket{\psi}_{LR}=1/\sqrt{d_L}\sum_{k=1}^{d_L}\ket{k}_{L}\ket{k}_{R}$ is the maximally entangled state. As a special case, we take $\Lambda'$ as the identity map and we obtain the entanglement fidelity of the channel $\Lambda$,
\begin{equation}\label{eq:enfide}
\begin{split}
F_e(\Lambda)&=F_e(\Lambda,\mathcal{I})\\
&=\sqrt{\bra{\psi}_{LR}(\Lambda_L\otimes \mathcal{I}_R)(\ketbra{\psi}_{LR})\ket{\psi}_{LR}}.  
\end{split}
\end{equation}

With the above entanglement fidelity, now we can characterize the performance of an encoding channel $\mathcal{E}$ under noise $\mathcal{N}$ by the quantity defined as
\begin{equation}\label{eq:encha}
   f_e(\mathcal{N}\circ\mathcal{E}) =\underset{\mathcal{R}}{\max}F_e(\mathcal{R} \circ\mathcal{N}\circ\mathcal{E}).
\end{equation}
 When $f_e(\mathcal{N}\circ\mathcal{E})=1$, we can find a channel $\mathcal{R}$ such that all states are recovered perfectly.

The maximization problem in Eq.~\eqref{eq:encha} is generally difficult since the optimization is over all channels. Fortunately, we can study the problem from the view of leaking information to the environment via the method of complementary channels \cite{PhysRevLett.104.120501}. As shown in Fig.~\ref{figure}, the channel $(\mathcal{N}\circ\mathcal{E})_{L\rightarrow S}$ has an isometry dilation $V_{L\rightarrow SE}$ with environment system $E$ such that
\begin{equation}
   \mathcal{N}\circ\mathcal{E}(\rho_L)=\tr_{E}\Big(V_{L\rightarrow SE}\rho_L V_{L\rightarrow SE}^{\dagger}\Big).
\end{equation}
Then the complementary channel is defined as
\begin{equation}
    \widehat{\mathcal{N}\circ\mathcal{E}}(\rho_L)=\tr_{S}\Big(V_{L\rightarrow SE}\rho_L V_{L\rightarrow SE}^{\dagger}\Big).
\end{equation}
The optimization problem \eqref{eq:encha} has an equivalent form \cite{PhysRevX.10.041018}
\begin{equation}\label{eq:com}
    f_e(\mathcal{N}\circ\mathcal{E})
    =\underset{\ket{\zeta}}{\max}F_e(\widehat{\mathcal{N}\circ\mathcal{E}},\mathcal{T}_{\zeta}),
\end{equation}
where $\mathcal{T}_{\zeta}(\cdot)=\tr(\cdot)\ketbra{\zeta}$ is a constant channel. 

\begin{figure}[hbtp!]
	\includegraphics[scale=0.25]{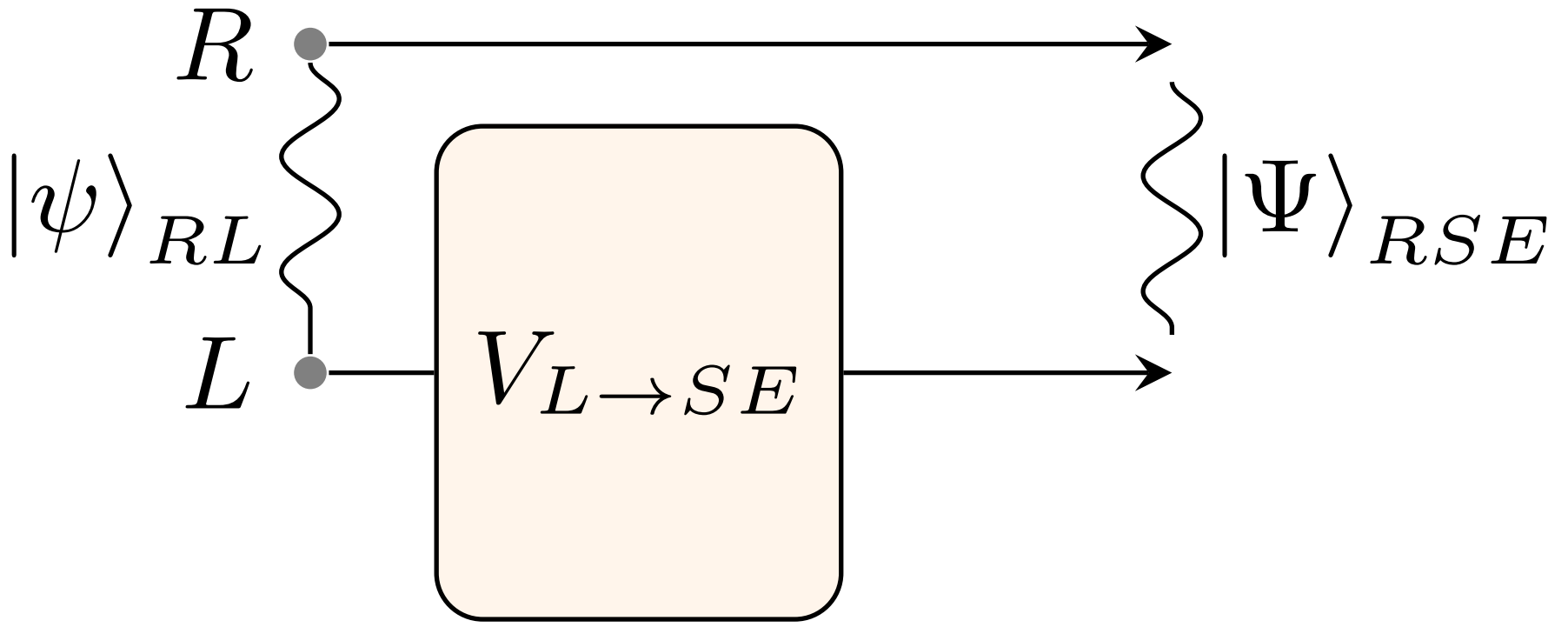}
	\caption{The $V_{L\rightarrow SE}$ is a Stinespring dilation of the channel $(\mathcal{N}\circ \mathcal{E})_{L\rightarrow S}$ with an environment system $E$. The input is a maximally entangled state of the logical system $L$ and the reference system $R$. The output is denoted by $\ket{\Psi}_{RSE}$.}
 \label{figure}
\end{figure}

With the method of complementary channels, we give a lower bound of the entanglement fidelity $f_e$ for generalized encoding and noise channels which extends the results in Ref. \cite{klesse2007approximate}. 
\begin{lemma}
Suppose the channel $(\mathcal{N}\circ\mathcal{E})_{L\rightarrow S}$ has a Stinespring dilation $V_{L\rightarrow SE}$, as shown in Fig.~\ref{figure}. Let
\begin{equation}
    \ket{\Psi}_{SER}=(V_{L\rightarrow SE}\otimes \mathbf{1}_{R})\ket{\psi}_{LR},
\end{equation}
and denote by $\rho_{RE}$, $\rho_R$ and $\rho_E$ the reduced states of $\ket{\Psi}_{RSE}$ on $RE$, $R$, and $E$, respectively. The quantity $f_e$ satisfies the inequality
\begin{equation}\label{eq:lem}
\begin{split}
1-f_e(\mathcal{N}\circ\mathcal{E})&\leq \frac{1}{2}\norm{\rho_{RE} -\rho_R\otimes\rho_E}_{1}\\
&\leq \frac{\sqrt{d_L d_E}}{2}\norm{\rho_{RE} -\rho_R\otimes\rho_E}_{2},
\end{split}
\end{equation}
where $d_E$ is the dimension of the environment system $E$. The equality holds if and only if $\rho_{RE}=\rho_{R}\otimes \rho_E$.
\end{lemma}

\begin{proof}
From Eq.~\eqref{eq:com}, we obtain
\begin{equation}
    \begin{split}
        &f_e(\mathcal{N}\circ\mathcal{E})\\
        &=\underset{\ket{\zeta}}{\max}F(\widehat{\mathcal{N}\circ\mathcal{E}}\otimes\mathcal{I}(\ketbra{\psi}_{LR}), \mathcal{T}_{\zeta}\otimes\mathcal{I}(\ketbra{\psi}_{LR}))\\
        &=\underset{\ket{\zeta}}{\max}F(\rho_{RE},\frac{\mathbf{1}_R}{d_L}\otimes \ketbra{\zeta}_E)\\
        &=\underset{\ket{\zeta}}{\max}F(\rho_{RE},\rho_R\otimes \ketbra{\zeta}_E)\\
        &\geq F(\rho_{RE},\rho_{R}\otimes \rho_E)\\
        &\geq 1-\frac{1}{2}\norm{\rho_{RE}-\rho_{R}\otimes \rho_E}_{1},
    \end{split}
\end{equation}
where the last inequality is from Eq.~\eqref{eq:ineqf}. 

Recall that for an operator $A$, the $1$-norm and the $2$-norm have the relation 
\begin{equation}
    \norm{A}_1\leq \sqrt{{\rm rank}(A)}\norm{A}_2.
\end{equation}
According to this relation and ${\rm rank}(\rho_{RE}-\rho_{R}\otimes \rho_E)\leq d_L d_E$, we obtain the remaining inequality in Eq.~\eqref{eq:lem}.
\end{proof}

In general, the fidelity and the $1$-norm are difficult to calculate since we need spectral decomposition. In comparison, the $2$-norm is easier to calculate. After a tedious calculation of the $2$-norm in Eq.~\eqref{eq:lem} presented in Appendix A, we obtain a lower bound of $f_e$.
\begin{observation}\label{ob:low}
Let $\mathcal{E}_{L\rightarrow S}$ and $\mathcal{N}_{S\rightarrow S}$ be the encoding and noise channels, respectively. Suppose they have specific forms,
\begin{equation}
\begin{split}
    \mathcal{E}(\rho)&=\sum_{s=1}^{m}E_{s}\rho E_{s}^{\dagger},\\
    \mathcal{N}(\sigma)&=\sum_{i=1}^{n}A_{i}\sigma A_{i}^{\dagger},
\end{split}
\end{equation}
where $\sum_{s=1}^{m}E_s^{\dagger} E_s=\mathbf{1}_L$ and $\sum_{i=1}^{n}A_i^{\dagger} A_i=\mathbf{1}_S$. Define $O=\sum_{s=1}^{m}E_s E_s^{\dagger}$. The entanglement fidelity $f_e$ has a lower bound,
\begin{equation}
    f_e(\mathcal{N}\circ\mathcal{E})\geq 1-\epsilon(\mathcal{N}\circ\mathcal{E}),
\end{equation}
where
\begin{equation}
\begin{split}
    \epsilon(\mathcal{N}\circ\mathcal{E})&=\sqrt{\frac{mn}{4 d_L}}\Big(\sum_{i,j=1}^{n}\tr(A_i^{\dagger}A_j O A_j^{\dagger} A_i O)\\
    &-\frac{1}{d_L}\sum_{i,j=1}^{n}\sum_{s,t=1}^{m}\abs{\tr(A_i^{\dagger}A_jE_t E_s^{\dagger})}^2\Big)^{1/2}
\end{split}
\end{equation} 
and we call $\epsilon$ the infidelity.
\end{observation}

This observation gives a quantitative description of the performance of an approximate QEC. When $\epsilon \ll 1$, the errors can be corrected approximately. The defined infidelity $\epsilon$ also characterizes the correlation between system $R$ and system $E$. As the environment becomes more correlated with the reference system which contains the encoded quantum information, more information leaks into the environment, which can result in the degradation of the protected information. 

\section{Covariance symmetry}\label{sec:nonc}
A channel $\mathcal{E}$ from system $L$ to system $S$ is called covariant with group $\mathbf{G}$, if for all $ g\in \mathbf{G}$ and all $ \rho \in \mathcal{D}(\mathcal{H}_L)$ there is
\begin{equation}
    \mathcal{E}\Big(U_L(g)\rho U_L(g)^{\dagger}\Big)=U_S(g)\mathcal{E}(\rho) U_S^{\dagger}(g),
\end{equation}
where $U_L(g)$ and $U_S(g)$ are unitary representations of group element $g$ on space $\mathcal{H}_L$ and $\mathcal{H}_S$, respectively. We can also say that the channel is symmetric with respect to $\mathbf{G}$. The covariant channel is intimately connected with the symmetric state and the Choi representation builds this bridge. More explicitly, the covariance symmetry of a channel is equal to the group symmetry of the corresponding Choi state \cite{d2004extremal}. Now we explain this equivalence relation in detail.

Recall that there exists a one-to-one correspondence between the channel and the Choi state
\begin{equation}
    \begin{split}
        \Phi_{\mathcal{E}}&=(\mathcal{I}_{L}\otimes\mathcal{E}_{R\rightarrow S})(\ketbra{\psi}_{LR}),\\
        \mathcal{E}_{L\rightarrow S}(\rho)&= d_L \tr_{L}[(\rho_L^T \otimes\mathbf{1}_S)\Phi_{\mathcal{E}}],
    \end{split}
\end{equation}
where $T$ represents the transposition.
Suppose the channel $\mathcal{E}$ is $\mathbf{G}$ covariant. Then for all $ \rho $ and all $ g $, we can obtain
\begin{equation}
    \begin{split}
        0=&\frac{1}{d_L}\mathcal{E}(\rho)-\frac{1}{d_L}U_S^{\dagger}(g)\mathcal{E}\Big(U_L(g)\rho U_L^{\dagger}(g)\Big)U_S(g)\\
        =& \tr_{L}[(\rho_L^T\otimes \mathbf{1}_S)\Phi_{\mathcal{E}}]\\
        -&  U_S^{\dagger}(g) \tr_{L}\{[ U_L^{*}(g)\rho_L^T U_L^{T}(g)\otimes \mathbf{1}_S]\Phi_{\mathcal{E}}\} U_{S}(g),
    \end{split}
\end{equation}
where $*$ is the conjugate operation. Therefore, 
\begin{equation}
    \Phi_{\mathcal{E}}=[U_L^{T}(g)\otimes U_S^{\dagger}(g)]\Phi_{\mathcal{E}}[U_L^{*}(g)\otimes U_S(g)] \forall g.
\end{equation}
This implies that the Choi state $\Phi_{\mathcal{E}}$ is symmetric with respect to the unitary representation $\{U_L^{*}(g)\otimes U_S(g)\}$. Hence, we can quantify the noncovariance of a channel from the asymmetry of its Choi state. Concretely, noncovariance of the channel $\mathcal{E}$ is defined as
\begin{equation}\label{eq:covmea}
    N_{\mathbf{G}}(\mathcal{E})=N_{\mathbf{G}}(\Phi_{\mathcal{E}})=\sum_{p=1}^{d_{\mathbf{G}}}I(\Phi_{\mathcal{E}},H_p),
\end{equation}
which has been thoroughly studied as discussed in Sec. \ref{sec:pre}.

\section{Isometric encoding}\label{sec:iso}
In this section we investigate infidelity, noncovariance, and their trade-off relation of a particular example.
\subsection{Infidelity of the QEC}
The isometric encoding channel is of the form
$\mathcal{E}(\rho)=W\rho W^{\dagger}$, with $W^{\dagger}W=\mathbf{1}_{L}$. The projector onto the coding space is $P=WW^{\dagger}$ and the Choi state is
\begin{equation}
\begin{split}
  \Phi_{\mathcal{E}}&=(\mathcal{I}_{L}\otimes\mathcal{E}_{R\rightarrow S})(\ketbra{\psi}_{LR})\\
    &=\ketbra{\Tilde{\psi}}_{LS},   
\end{split}
\end{equation}
where $\ket{\Tilde{\psi}}_{LS}=(\mathbf{1}_L\otimes W_{R\rightarrow S})\ket{\psi}_{LR}$.
The noise channel is assumed to be of the general form $\mathcal{N}_{S\rightarrow S}(\rho)=\sum_{i=1}^{n}A_i \rho A_i^{\dagger}$.

According to Observation \ref{ob:low}, the square of infidelity is

\begin{equation*}
\begin{split}
    &\epsilon^2(\mathcal{N}\circ\mathcal{E})\\
    &=\frac{n}{4 d_L}\sum_{i,j=1}^{n}\tr(P A_j^{\dagger} A_i P A_i^{\dagger} A_j )-\frac{1}{d_L}\abs{\tr (P A_i^{\dagger}A_j)}^2.   
\end{split}
\end{equation*}
Notice that when the Knill-Laflamme condition is satisfied, namely, $P(A_i^{\dagger} A_j )P=\lambda_{ij}P$ holds for all $i$ and $j$ with some constants $\lambda_{ij}$, then infidelity $\epsilon=0$ and perfect error correction can be realized.

We define $K_{ij}=P A_i^{\dagger}A_j P$ and the infidelity can be written in the form of the generalized skew information
\begin{equation*}
    \begin{split}
       &\frac{4d_L \epsilon^2(\mathcal{N}\circ \mathcal{E})}{n} \\
       &=\sum_{i,j=1}^{n}\frac{1}{2}\tr( K_{ij}^{\dagger} K_{ij}+K_{ij} K_{ij}^{\dagger} )-\sum_{i,j=1}^{n}\frac{1}{d_L}\abs{\tr K_{ij}}^2\\
       &= \sum_{i,j=1}^{n}\frac{1}{2}\tr(W^{\dagger} (K_{ij}^{\dagger} K_{ij}+K_{ij} K_{ij}^{\dagger}) W)\\
       &-\sum_{i,j=1}^{n}\frac{1}{d_L}\abs{\tr(W^{\dagger} K_{ij}W)}^2\\
         &= \sum_{i,j=1}^{n}\sum_{k,l=1}^{d_L}\frac{1}{2}\bra{k}W^{\dagger}(K_{ij}^{\dagger} K_{ij}+K_{ij} K_{ij}^{\dagger}) W\ket{l}\braket{k}{l}\\
         &-\sum_{i,j=1}^{n}\abs{ \sum_{k,l=1}^{d_L}\bra{k}W^{\dagger}K_{ij} W\ket{l}\braket{k}{l}}^2\\
         &=d_L\sum_{i,j=1}^{n} \Big(\frac{1}{2}\bra{\psi}\mathbf{1}_L\otimes W^{\dagger} (K_{ij}^{\dagger}K_{ij}+ K_{ij} K_{ij}^{\dagger})\mathbf{1}_L\otimes W\ket{\psi}\\
         &-\abs{\bra{\psi}(\mathbf{1}_L\otimes W^{\dagger})(\mathbf{1}_L \otimes K_{ij})(\mathbf{1}_L\otimes W)\ket{\psi}}^2\Big)\\
         &=d_L \sum_{i,j=1}^{n}I\Big(\ketbra{\Tilde{\psi}},\mathbf{1}_L \otimes K_{ij}\Big).
    \end{split}
\end{equation*}

\subsection{trade-off relation between infidelity and noncovariance}
Consider a general Lie group $\mathbf{G}$ and denote the Lie algebra by $\mathcal{L}_{\mathbf{G}}$. Here $\{U_L^*(g)\}$ is a unitary representation of $\mathbf{G}$ in the logical space $\mathcal{H}_L$ and we assume that the associated representation of the Lie algebra is $\pi_L$, namely, $\{\pi_L(X):X \in\mathcal{L}_{\mathbf{G}}\}$ is the set of generators for $\{U_L^*(g)\}$. Similarly, for the unitary representation $\{U_S(g)\}$ in the physical system, suppose the associated representation of the Lie algebra is $\pi_S$. Hence, in the Hilbert space $\mathcal{H}_L\otimes\mathcal{H}_S$, $\{\pi_L(X)\otimes \mathbf{1}_S+\mathbf{1}_L\otimes \pi_S(X) :X \in\mathcal{L}_{\mathbf{G}}\}$ gives the representation of the Lie algebra with respect to the unitary representation $\{U_L^*(g)\otimes U_S(g)\}$ \cite{hall2013lie}, and we assume that the set $\{H_L^p\otimes \mathbf{1}_S + \mathbf{1}_L\otimes H_S^p:p=1,\cdots,d_{\mathbf{G}}\}$ constitutes an orthonormal base of the Lie algebra $\mathcal{L}_{\mathbf{G}}$. The sum of skew information
\begin{equation}
    N_{\mathbf{G}}(\rho)=\sum_{p=1}^{d_{\mathbf{G}}} I(\rho,H_L^p\otimes\mathbf{1}_S + \mathbf{1}_L\otimes H_S^p),
\end{equation}
quantifies the asymmetry of state $\rho$ with respect to $\mathbf{G}$ \cite{Li_2020}. We can obtain the noncovariance measure of a channel $\mathcal{E}$ from this asymmetry measure, as defined in Eq.~\eqref{eq:covmea}. Moreover, we find the following relation by combining Lemma \ref{lem} with the expressions of infidelity and noncovariance.
\begin{observation}
For an isometric encoding channel $\mathcal{E}$ and noise channel $\mathcal{N}$, noncovariance with respect to a Lie group $\mathbf{G}$ and infidelity satisfy the trade-off relation
\begin{equation}
   \frac{4 \epsilon^2(\mathcal{N}\circ \mathcal{E})}{n}+ N_{\mathbf{G}}(\mathcal{E})
   \geq \frac{1}{n^2+d_{\mathbf{G}}} I\Big(\ketbra{\Tilde{\psi}},K\Big),
\end{equation}
where $K=\sum_{i,j=1}^{n}\mathbf{1}_L\otimes K_{ij}+\sum_{p=1}^{d_{\mathbf{G}}}(H_L^p\otimes\mathbf{1}_S + \mathbf{1}_L\otimes H_S^p)$.
\end{observation}

Next we consider the special case of the U(1) group. In this case, we assume that $U_L(g)=e^{-i H_L^* g}$ and $U_S(g)=e^{-i H_S g}$, where $H_L^*$ and $H_S$ are Hamiltonians. Then
\begin{equation}
\begin{split}
    U_L^*(g)\otimes U_S(g)&=e^{i H_L g}\otimes e^{-i H_S g}\\
    &=e ^{-i(\mathbf{1}_L \otimes H_S-H_L\otimes\mathbf{1}_S)g}.
\end{split}
\end{equation}
 Thus, the corresponding generated Hamiltonian of $U_L^*(g)\otimes U_S(g)$ is $H =\mathbf{1}_L \otimes H_S-H_L\otimes\mathbf{1}_S$.
The noncovariance of the isometric encoding channel can be quantified by the skew information
\begin{equation}
   N_{\mathbf{G}}(\mathcal{E})=I\Big(\ketbra{\Tilde{\psi}},H\Big).
\end{equation}

For the U(1) group, the HKS condition is sufficient for the nonexistence of covariant and exact QEC codes \cite{PRXQuantum.2.010343}. Explicitly, if 
\begin{equation}
    H_S\in {\rm span}\{A_i^{\dagger}A_j:i,j=1,\cdots,n\},
\end{equation}
all covariant codes cannot correct errors perfectly. Here we prove again this no-go result. 

Let $H_S=\sum_{i,j=1}^{n}\alpha_{ij}A_i^{\dagger}A_j$, with $\alpha_{ij}\in \mathbb{C}$, and suppose the isometric encoding channel $\mathcal{E}$ is covariant and corrects errors perfectly. Since $N_{\mathbf{G}}(\mathcal{E})=0$ and $H$ is Hermitian, there exists a constant $\lambda$ such that
\begin{equation}
    H(\mathbf{1}_L\otimes W)\ket{\psi}=\lambda (\mathbf{1}_L\otimes W)\ket{\psi}.
\end{equation}
This implies that
\begin{equation}
     (\mathbf{1}_L\otimes P)H (\mathbf{1}_L\otimes P)(\mathbf{1}_L\otimes W)\ket{\psi}=\lambda (\mathbf{1}_L\otimes W)\ket{\psi}.
\end{equation}
Consequently, 
\begin{equation}\label{eq:e1}
    \begin{split}
    0&=I\Big(\ketbra{\Tilde{\psi}},(\mathbf{1}_L\otimes P)H(\mathbf{1}_L\otimes P)\Big)\\
    &=I\Big(\ketbra{\Tilde{\psi}},\sum_{i,j}\alpha_{ij}\mathbf{1}_L\otimes K_{ij}-H_L\otimes P\Big).
    \end{split}
\end{equation}
In addition, $\epsilon(\mathcal{N}\circ\mathcal{E})=0$ indicates that 
\begin{equation}\label{eq:e2}
    I(\ketbra{\Tilde{\psi}},\alpha_{ij}\mathbf{1}_L\otimes K_{ij})=0.
\end{equation} 
Combining Eq.~\eqref{eq:e1} with Eq.~\eqref{eq:e2}, we obtain 
\begin{equation}
\begin{split}
    0&\leq (n^2+1) I\Big(\ketbra{\Tilde{\psi}},H_L\otimes P\Big)
   \\ &\leq  I\Big(\ketbra{\Tilde{\psi}},-\sum_{i,j=1}^{n}\alpha_{ij}\mathbf{1}_L\otimes K_{ij}+H_L\otimes P\Big)\\
   &+\sum_{i,j=1}^{n}I\Big(\ketbra{\Tilde{\psi}},\alpha_{ij}\mathbf{1}_L\otimes K_{ij}\Big)=0.
\end{split}
\end{equation}
Therefore,
\begin{equation}\label{eq:a}
   (\mathbf{1}_L\otimes W^{\dagger}) (H_L\otimes P)(\mathbf{1}_L\otimes W)\ket{\psi}=\alpha \ket{\psi}
\end{equation}
holds for some constant $\alpha$. After a direct calculation, we have
\begin{equation}
    \bra{k}H_L\ket{l}=\alpha \delta_{kl},
\end{equation}
or equivalently, $H_L=H_L^*=\alpha \mathbf{1}_L$. This contradicts the nontrivial assumption of a logical Hamiltonian.

\subsection{Average infidelity and noncovariance for random codes}
We consider a type of random code in which the encoding isometry has the following expression
\begin{equation}\label{eq:randu}
    W=U_S(\mathbf{1}_L\otimes \ket{0}_A),
\end{equation}
where $A$ is an ancillary system satisfying $\mathcal{H}_S=\mathcal{H}_L\otimes \mathcal{H}_A$ and $U$ is a random unitary under the Haar measure. Equivalently, the projector can be written as
\begin{equation}
    P=U_S(\mathbf{1}_L\otimes \ketbra{0}_A)U_S^{\dagger}.
\end{equation}
For this type of random code, the average infidelity satisfies
\begin{equation}\label{eq:avgi}
\begin{split}
   & \int_{\mathbf{U}(d_S)}\epsilon^2(\mathcal{N}\circ\mathcal{E})d\mu(U)\\
    &=\frac{n (d_L^2-1)}{4d_L(d_S^2-1)} \\
    &\times \sum_{i,j=1}^{n}\Big(\tr(A_j^{\dagger}A_iA_i^{\dagger} A_j)-\frac{1}{d_S}\abs{\tr(A_i^{\dagger} A_j)}^2\Big),
\end{split}
\end{equation}
where $\mathbf{U}(d_S) $ represents the unitary group in system $S$ and $\mu$ is the Haar measure.
When $\mathbf{G}$ is the U(1) group, the average noncovariance is equal to
\begin{equation}\label{eq:avgn}
\begin{split}
  &\int_{\mathbf{U}(d_S)}  N_{\mathbf{G}}(\mathcal{E}) d\mu(U)\\
  &=\frac{d_L\tr (H_L^2)-(\tr H_L)^2}{d_L^2}\\
  &+\frac{(d_Ld_S^2-d_S)\tr(H_S^2)-(d_Ld_S-1)(\tr H_S)^2}{d_Ld_S(d_S^2-1)}.
\end{split}
\end{equation}
We leave the detailed calculation to Appendix B.

From Eqs.~\eqref{eq:avgi} and \eqref{eq:avgn}, we can see that if the dimension of the physical system $d_S$ tends to infinity, the average infidelity tends $0$ while the noncovariance tends to $[d_L\tr (H_L^2)-(\tr H_L)^2]/d_L^2$. 
\section{Conclusion and outlook}\label{sec:out}
In this work we defined a quantity termed infidelity to characterize the inaccuracy of an approximate QEC and also to quantify the noncovariance of an encoding channel with respect to a general Lie group. With these two quantities, we derived a trade-off relation between approximate QEC and noncovariance in the special case that the encoding channel is isometric. For a type of random code, we found that when the dimension of the physical system is large enough, the errors can be corrected approximately while noncovariance tends to a constant.

The information scrambling can protect encoding information against errors and hence is closely connected with the capability of error correction \cite{PhysRevLett.125.030505,PhysRevResearch.2.043164,PhysRevA.104.012408}. For future work it would be interesting to explore the QEC ability and the information scrambling quantitatively in systems with particular symmetry via the infidelity we defined, which may help us design explicit covariant and approximate QEC codes from scrambling circuits.
\acknowledgments
This work was supported by the National Natural Science Foundation of China Grant No.~12174216.

\appendix
\onecolumngrid
\section{Proof of Observation \ref{ob:low}}

The Stinespring isometry $V_{L\rightarrow SE}$ of the composite channel 
\begin{equation}
    \mathcal{N}\circ\mathcal{E}(\rho)=\sum_{i,s}A_i E_s\rho E_s^{\dagger}A_i^{\dagger},
\end{equation}
satisfies 
\begin{equation}
    V_{L\rightarrow SE}\ket{\varphi}_{L}=\sum_{i,s} A_i E_s\ket{\varphi}_L\otimes\ket{is}_E.
\end{equation}
Here $\{\ket{is}_E\}$ forms an orthonormal basis of the environment system $E$ and the dimension $d_E=mn$, which is equal to the number of the Kraus operators $\{A_i E_s\}$. Note that we omit the upper bound of the index in the summation sign for convenience in this appendix.

The output state is 
\begin{equation}
\begin{split}
    \ket{\Psi}_{RSE}&=(\mathbf{1}_R\otimes V_{L\rightarrow SE})\ket{\psi}_{RL}\\
    &=\frac{1}{\sqrt{d_L}}\sum_{k} V_{L\rightarrow SE}\ket{k}_L \otimes \ket{k}_R\\
    &=\frac{1}{\sqrt{d_L}}\sum_{k,i,s} A_i E_s\ket{k}_L \otimes \ket{is}_E \otimes \ket{k}_R.
\end{split}
\end{equation}

The reduced state in system $RE$ is
\begin{equation}
    \begin{split}
        \rho_{RE}=&\tr_S \ketbra{\Psi}_{RSE} \\
        =&\frac{1}{d_L}\sum_{i,j,k,l,s,t}\tr_{S}[(A_i E_s\ketbra{k}{l}E_t^{\dagger}A_j^{\dagger})_{S}\otimes \ketbra{is}{jt}_E\otimes\ketbra{k}{l}_R]\\
        =&\frac {1}{d_L}\sum_{i,j,k,l,s,t}\bra{l}E_t^{\dagger} A_j^{\dagger}A_i E_s\ket{k} \ketbra{is}{jt}_E\otimes\ketbra{k}{l}_R.
    \end{split}
\end{equation}
Then the reduced state in system $R$ is the maximally mixed state $\mathbf{1}_{R}/ d_L$ and 
\begin{equation}
    \begin{split}
        \rho_E&=\tr_{R}\rho_{RE}\\
        &=\frac{1}{d_L}\sum_{k,i,j,s,t}\bra{k}E_t^{\dagger}A_j^{\dagger}A_i E_s \ket{k}\ketbra{is}{jt}_E\\
        &=\frac{1}{d_L}\sum_{i,j,s,t}\tr(E_t^{\dagger}A_j^{\dagger}A_i E_s)\ketbra{is}{jt}_E.
    \end{split}
\end{equation}

To calculate the $2$-norm in Eq.~\eqref{eq:lem}, we map the states in system $RE$ to states in system $LE$ through the isometric channel
\begin{equation}
\begin{split}
    &\Lambda(\sum_{k,l,i,j,s,t}\alpha_{klijst}\ketbra{is}{jt}_{E} \otimes \ketbra{k}{l}_{R})\\
    &=\sum_{k,l,i,j,s,t}\alpha_{klijst}^{*}\ketbra{is}{jt}_{E} \otimes \ketbra{k}{l}_{L},
\end{split}
\end{equation}
and then
\begin{equation}
    \norm{\Lambda(\rho_{RE})-\Lambda(\rho_R\otimes \rho_E)}_2=\norm{\rho_{RE}-\rho_R\otimes \rho_E}_2.
\end{equation}

If we let 
\begin{equation}
    \begin{split}
        D&=\Lambda(\rho_{RE})-\Lambda(\rho_R\otimes \rho_E)\\
        &=\frac{1}{d_L}\sum_{i,j,s,t}\Big(E_s^{\dagger}A_i^{\dagger} A_j E_t-\tr(E_s^{\dagger}A_i^{\dagger} A_j E_t)\frac{\mathbf{1}_L}{d_L}\Big)\otimes \ketbra{is}{jt}_E,
    \end{split}
\end{equation}
then we have 
\begin{equation}
    \begin{split}
        \norm{\rho_{RE}-\rho_R\otimes \rho_E}_2^2&=\tr DD^{\dagger}\\
        &=\frac{1}{d_L^2}\sum_{i,j,s,t}\Bigg(\tr(E_s^{\dagger}A_i^{\dagger} A_j E_t E_t^{\dagger} A_j^{\dagger} A_i E_s)-\frac{1}{d_L}\tr(E_s^{\dagger}A_i^{\dagger} A_j E_t)\tr(E_t^{\dagger} A_j^{\dagger} A_i E_s)\Bigg)\\
        &=\frac{1}{d_L^2}\sum_{i,j}\tr(A_i^{\dagger} A_j O A_j^{\dagger} A_i O)-\frac{1}{d_L^3}\sum_{i,j,s,t}\abs{\tr(A_i^{\dagger} A_j E_t E_s^{\dagger})}^2,
    \end{split}
\end{equation}
where $O=\sum_{t}E_t E_t^{\dagger}$.

\section{Calculation of average infidelity and noncovariance}
In a Hilbert space $\mathcal{H}$ with dimension $d$, the uniform Haar measure $\mu$ over unitary operator group $\mathbf{U}(d)$ remains invariant under both left and right multiplication of any unitary operator $V\in \mathbf{U}(d)$ \cite{zhang2014matrix,collins2016random,roberts2017chaos}. Mathematically,
\begin{equation}
    \mu(\mathcal{A})=\mu(\mathcal{A}V)=\mu(V\mathcal{A})
\end{equation}
holds for an arbitrary Borel subset $\mathcal{A}$ and arbitrary unitary $V$. Here we recall some integral formulas over unitary groups, referring to Ref. \cite{zhang2014matrix} for detailed proofs.
\begin{lemma}\label{lem:harfor}
For Haar measure $\mu$, it holds that
\begin{enumerate}
\item \begin{equation}
   \int_{\mathbf{U}(d)} U A U^{\dagger} d \mu (U)=\frac{\tr A}{d} \mathbf{1}_d,
 \end{equation} 
\item\begin{equation}
    \int_{\mathbf{U}(d_A)}(U_A\otimes \mathbf{1}_B) X_{AB}(U_A\otimes \mathbf{1}_B )^{\dagger} d\mu(U_A)=\frac{\mathbf{1}_A}{d_A}\otimes \tr_A X_{AB},
\end{equation}
\item\begin{equation}
    \begin{split}
    &\int_{\mathbf{U}(d)}(U\otimes U)A(U\otimes U)^{\dagger}d\mu(U)\\
    &=\Big(\frac{\tr A}{d^2-1}-\frac{\tr(AF)}{d(d^2-1)}\Big)\mathbf{1}_{d^2}-\Big(\frac{\tr A}{d(d^2-1)}-\frac{\tr(AF)}{d^2-1}\Big)F,
    \end{split}
\end{equation}
where $F$ is the swap operator.
\item\begin{equation}
\begin{split}
    &\int_{\mathbf{U}(d)}UAU^{\dagger}XUBU^{\dagger} d\mu(U)\\
    &=\frac{d\tr(AB)-\tr A\tr B}{d(d^2-1)}(\tr X) \mathbf{1}_d+\frac{d\tr A\tr B-\tr(AB)}{d(d^2-1)}X.
\end{split}
\end{equation}
\end{enumerate}
\end{lemma}
We first calculate the average infidelity. According to the lemma \ref{lem:harfor}, we have
\begin{equation}
    \begin{split}
      &\int_{\mathbf{U}(d_S)}\tr (PA_j^{\dagger}A_i P A_i^{\dagger} A_j)d\mu(U)\\
      &=\tr\int_{\mathbf{U}(d_S)}U(\mathbf{1}_L\otimes \ketbra{0})U^{\dagger}A_j^{\dagger}A_i U(\mathbf{1}_L\otimes \ketbra{0})U^{\dagger}A_i^{\dagger} A_j d\mu(U)\\
      &=\frac{d_S d_L-d_L^2}{d_S(d_S^2-1)}\abs{\tr(A_j^{\dagger}A_i)}^2+\frac{d_Sd_L^2-d_L}{d_S(d_S^2-1)}\tr(A_j^{\dagger}A_iA_i^{\dagger} A_j)
    \end{split}
\end{equation}
and
\begin{equation}
    \begin{split}
        &\int_{\mathbf{U}(d_S)}\abs{\tr(PA_i^{\dagger} A_j)}^2d\mu(U)\\
        &=\tr\int_{\mathbf{U}(d_S)}U(\mathbf{1}_L\otimes \ketbra{0})U^{\dagger}A_i^{\dagger} A_j\otimes U(\mathbf{1}_L\otimes \ketbra{0})U^{\dagger}A_j^{\dagger}A_i d\mu(U)\\
        &=\frac{d_S d_L^2-d_L}{d_S(d_S^2-1)}\abs{\tr(A_i^{\dagger} A_j)}^2-\frac{d_L^2-d_S d_L}{d_S(d_S^2-1)}\tr(A_j^{\dagger}A_iA_i^{\dagger} A_j).
    \end{split}
\end{equation}
Thus, we can obtain
\begin{equation}
\begin{split}
   & \int_{\mathbf{U}(d_S)}\frac{4d_L\epsilon^2(\mathcal{N}\circ\mathcal{E})}{n}d\mu(U)\\
    &=\frac{d_L^2-1}{d_S^2-1}\sum_{i,j}\Big(\tr(A_j^{\dagger}A_iA_i^{\dagger} A_j)-\frac{1}{d_S}\abs{\tr(A_i^{\dagger} A_j)}^2\Big).
\end{split}
\end{equation}

Next, we calculate the average of noncovariance
\begin{equation}
N_{\mathbf{G}}(\mathcal{E})=\bra{\Tilde{\psi}}H^2\ket{\Tilde{\psi}}-\bra{\Tilde{\psi}}H\ket{\Tilde{\psi}}^2.
\end{equation}
The first term is equal to
\begin{equation}
    \begin{split}
      &\int_{\mathbf{U}(d_S)} \bra{\Tilde{\psi}}H^2\ket{\Tilde{\psi}} d\mu(U)\\
       &=\int_{\mathbf{U}_{d_S}}\bra{\psi'}(\mathbf{1}_L\otimes U)^{\dagger} H^2(\mathbf{1}_L\otimes U )\ket{\psi'}d\mu(U)\\
       &=\frac{1}{d_L}\tr (H_L^2)+\frac{1}{d_S}\tr (H_S^2)-\frac{2 \tr H_L \tr H_S}{d_L d_S},
    \end{split}
\end{equation}
where $\ket{\psi'}=1/\sqrt{d_L}\sum_{k}\ket{k}_L\ket{k0}_{S}$. 

The second term is equal to 
\begin{equation}
\begin{split}
    &\int_{\mathbf{U}(d_S)} \bra{\Tilde{\psi}}H\ket{\Tilde{\psi}}^2 d\mu(U)\\
    &=\frac{1}{d_L^2}\int_{\mathbf{U}(d_S)}\{-\tr H_L+\tr[U^{\dagger} H_S U(\mathbf{1}_L\otimes \ketbra{0})]\}^2 d\mu(U)\\
    &=\frac{(\tr H_L)^2}{d_L^2} -\frac{2\tr H_L}{d_L^2}\tr\int_{\mathbf{U}(d_S)}(U^{\dagger} H_S U)(\mathbf{1}_L\otimes \ketbra{0})d\mu(U)\\
   &+\frac{1}{d_L^2} \tr\int_{\mathbf{U}(d_S)}(U^{\dagger})^{\otimes 2}H_S^{\otimes 2}U^{\otimes 2}(\mathbf{1}_{L}\otimes \ketbra{0})^{\otimes 2}d\mu(U)\\
   &=\frac{(\tr H_L)^2}{d_L^2}-\frac{2\tr H_L\tr H_S}{d_L d_S}+\frac{ d_Ld_S-1}{d_Ld_S(d_S^2-1)}(\tr H_S)^2 +\frac{d_S-d_L}{d_Ld_S(d_S^2-1)}\tr(H_S^2)
\end{split}
\end{equation}
Thus, the average noncovariance can be expressed as
\begin{equation}
\begin{split}
  &\int_{\mathbf{U}(d_S)}  N_{\mathbf{G}}(\mathcal{E}) d\mu(U)\\
  &=\frac{d_L\tr (H_L^2)-(\tr H_L)^2}{d_L^2}+\frac{(d_Ld_S^2-d_S)\tr(H_S^2)-(d_Ld_S-1)(\tr H_S)^2}{d_Ld_S(d_S^2-1)}.
\end{split}
\end{equation}

\end{document}